\documentclass[12pt, draftclsnofoot, onecolumn]{IEEEtran}
\IEEEoverridecommandlockouts
\usepackage{graphicx}
\usepackage[numbers,sort&compress]{natbib}
\usepackage{amsthm}

\usepackage{amsmath}
\usepackage{amssymb}

\usepackage[space]{grffile}
\usepackage{subfig}
\usepackage[font=scriptsize]{caption}

\theoremstyle{remark}
\newtheorem{thm}{$\;\;\;$Theorem}
\newtheorem{lem}{$\;\;\;$Lemma}

\newtheorem{defn}{$\;\;\;$Definition}
\def\BibTeX{{\rm B\kern-.05em{\sc i\kern-.025em b}\kern-.08em
    T\kern-.1667em\lower.7ex\hbox{E}\kern-.125emX}}
\begin{document}

\title{Impact of Spatial Multiplexing on the Throughput of Ultra-Dense mmWave AP Networks
}

\author{Shuqiao~Jia and
	Behnaam~Aazhang~\IEEEmembership{}}

\maketitle

\begin{abstract}
The operating range of a single millimeter wave (mmWave) access point (AP) is small due to the high path loss and blockage issues of the frequency band. To achieve the coverage similar to conventional sub-6GHz networks, the ultra-dense deployments of APs are required by the mmWave network. In general, the mmWave APs can be categorized into backhaul-connected APs and relay APs. 
Though the spatial distribution of backhaul-connected APs can be captured by the Poison point process (PPP), the desired locations of relay APs depend on the transmission protocol operated in the mmWave network. In this paper, we consider modeling the topology of mmWave AP network by incorporating the multihop protocol. We first derive the topology of AP network with the spatial multiplexing disabled for each transmission hop. Then we analyze the topology when the spatial multiplexing is enabled at the mmWave APs. To derive the network throughput, we first quantify the improvement in latency and the degradation of coverage probability with the increase of spatial multiplexing gain at mmWave APs. Then we show the impact of spatial multiplexing on the throughput for the ultra-dense mmWave AP network.

\end{abstract}


\section{Introduction}
The use of millimeter wave (mmWave) frequencies in access points (APs) becomes a trend in the emerging fifth generation network \cite{andrews2014will,di2015stochastic,jia2016impact,andrews2017modeling,bai2015coverage,rappaport2013millimeter}. 
Despite the large available bandwidth in mmWave frequencies, the small wavelength experiences a high path loss and a severe penetration loss, which limits the coverage of a single mmWave AP \cite{rappaport2015wideband}.
%
To achieve the same size of network coverage as the sub-6GHz network, the ultra-dense deployment of APs appears to be the solution for mmWave network \cite{jia2016impact,bai2015coverage,alammouri2018sinr}. 
Note that the topology of sub-6GHz network is relatively simple, where all the APs are connected to the Internet backhaul and each cell is covered by one AP in the network \cite{andrews2011tractable}.
However, owing to its ultra-density, only a small portion of APs in the mmWave network have direct access to the Internet backhaul. Other mmWave APs are used as relays to extend the coverage for the network, as shown in Fig.\ref{Figure: sub-tier structure} \cite{jia2016impact}. Accordingly, the topology of mmWave AP network varies with different multihop transmission protocols, which lead to different network performances. 

Several aspects of the ultra-dense mmWave AP network have been studied.  
In \cite{andrews2017modeling}, the mmWave modeling was comprehensively studied.
The hybrid precoding for mmWave network was proposed in \cite{alkhateeb2015limited}.
In \cite{jia2016impact}, the optimal intensity of the ultra-dense mmWave AP network was derived under the impact of blockage. 
The SINR coverage probability and rate analysis for mmWave networks were presented in \cite{di2015stochastic,bai2015coverage}.
We remark that all the previous work assumed mmWave APs to be uniformly distributed in the mmWave network. Such an assumption has been validated for the sub-6GHz network \cite{andrews2011tractable}. However, the ultra-density of mmWave APs results in a complicated and flexible network topology, which cannot be captured by simply applying a uniformly distributed spatial model.

We approach the topology of ultra-dense mmWave AP network by introducing the tiered model, rather than modeling the network as a whole.
In the tiered mmWave AP network, the backhaul-connected APs are considered as the $0^{\text{th}}$ tier. 
Other AP tiers in the mmWave network are used as relay, which extend the coverage of backhaul-connected APs to the whole network. 
Note that the spatial distribution of backhaul-connected APs is determined by the infrastructure. However, the AP locations in other tiers are decided by the topological structures of previous tiers and the transmission protocol. Consequently, the topology of mmWave AP network has large flexibility and is dependent on the transmission protocol.

Our key contribution is to analyze the performance of ultra-dense mmWave AP network with considering the impact of transmission protocol. Specifically, we analyze how the hybrid precoding scheme affects the performance of ultra-dense mmWave AP network. 
Note that the implementation of hybrid precoding combines the analog beamforming and baseband spatial multiplexing \cite{alkhateeb2015limited}. In this paper, we focus on the impact of spatial multiplexing on the performance of mmWave AP networks.
In section II, we first introduce the tiered structure for the mmWave APs. We then derive the topology of mmWave AP network concerning the different spatial multiplexing gains in Section III. The performance analysis is provided in section IV, where we characterize the latency, the coverage probability and the throughput for the mmWave AP network. 

\section{Tiered Model of mmWave APs}
\begin{figure*}
	\vspace*{-20pt}
	\centering
	\begin{minipage}{1\columnwidth}
	\raggedleft	
	\subfloat[{\scriptsize } \label{Figure: Sub-teri1}]{%
		\includegraphics[width=0.25\linewidth]{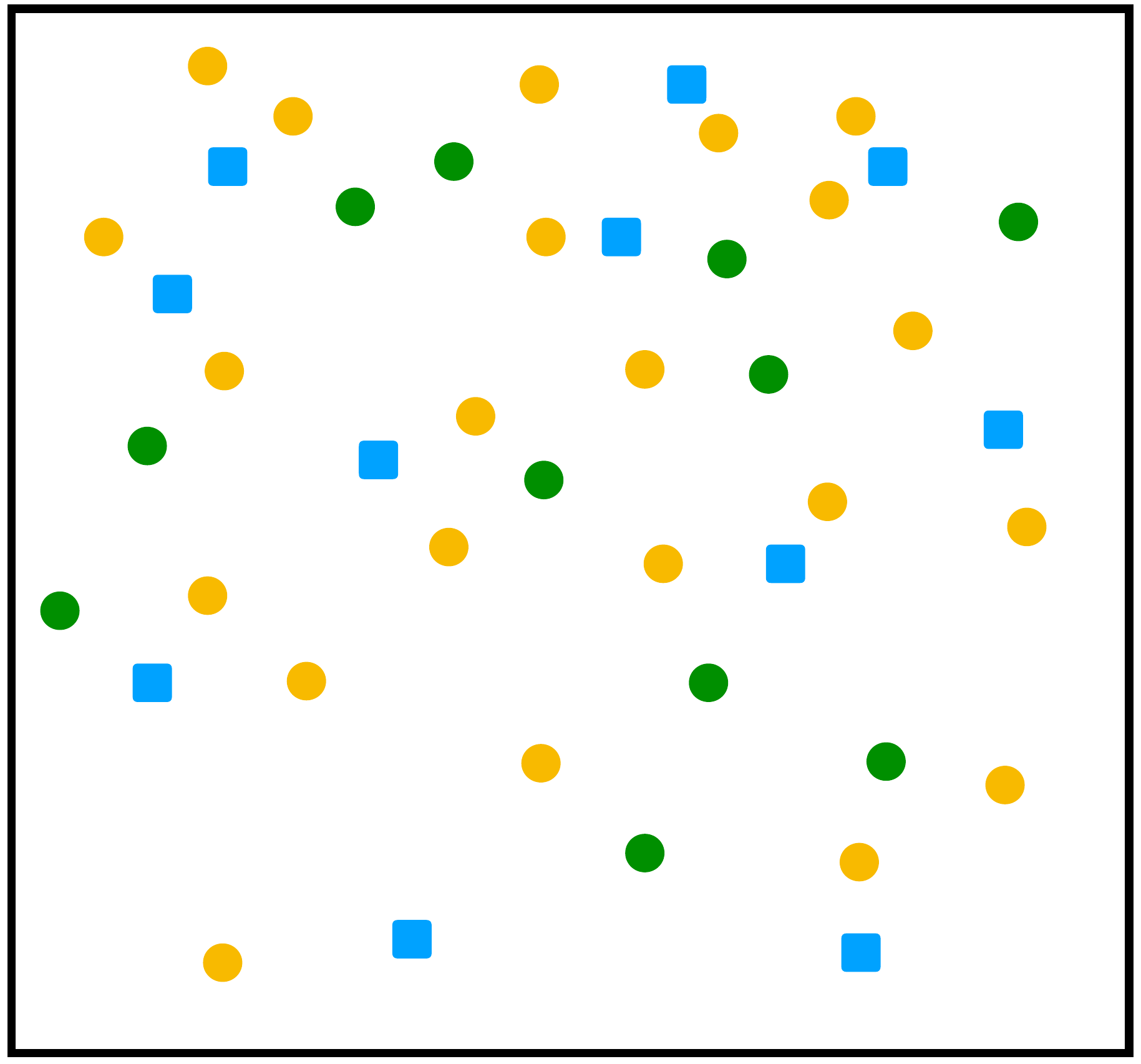}}
	\subfloat[{\scriptsize }\label{Figure: Sub-teri2}]{%
		\includegraphics[width=0.25\linewidth]{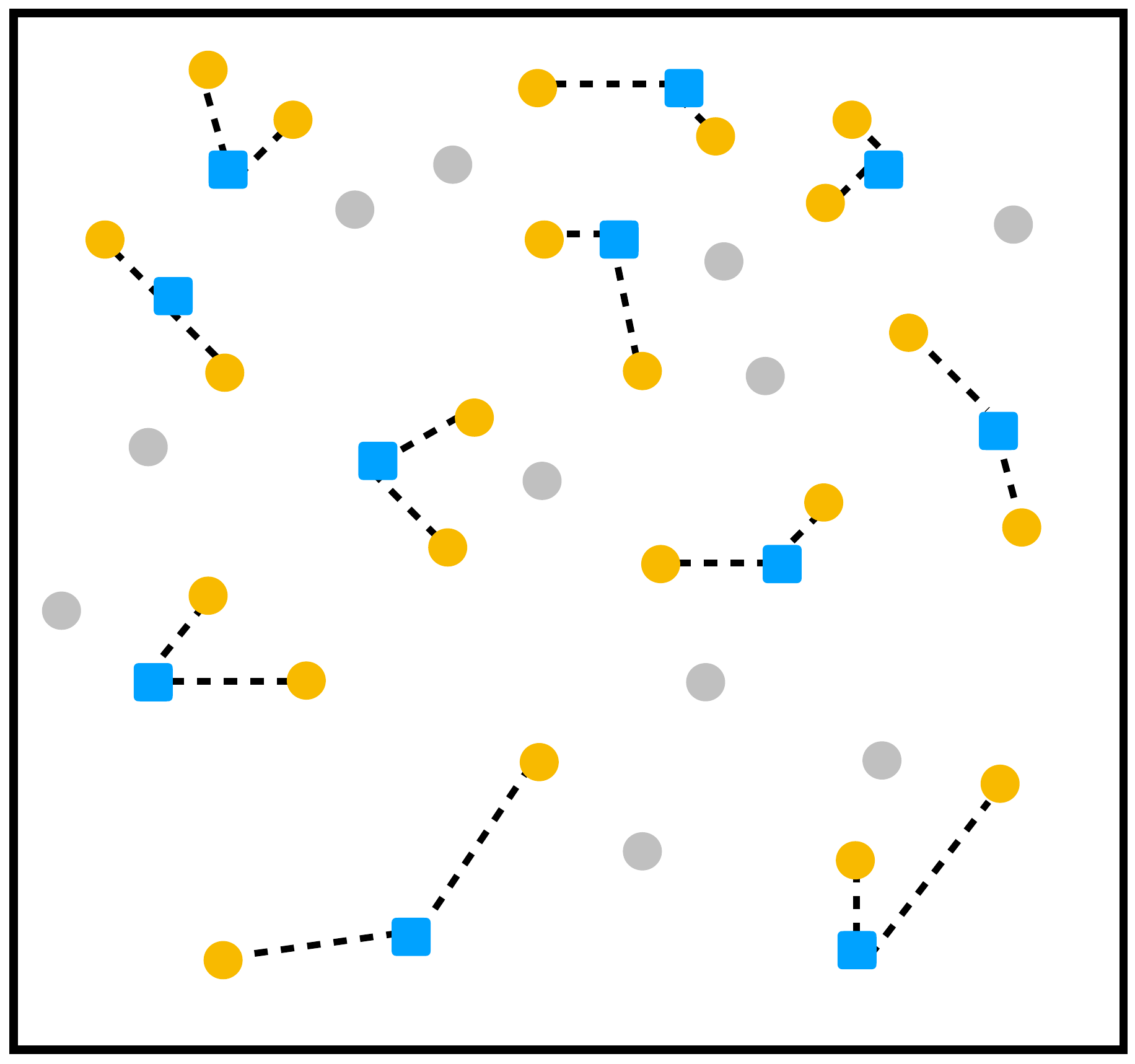}}
	\subfloat[{\scriptsize } \label{Figure: Sub-teri3}]{%
		\includegraphics[width=0.25\linewidth]{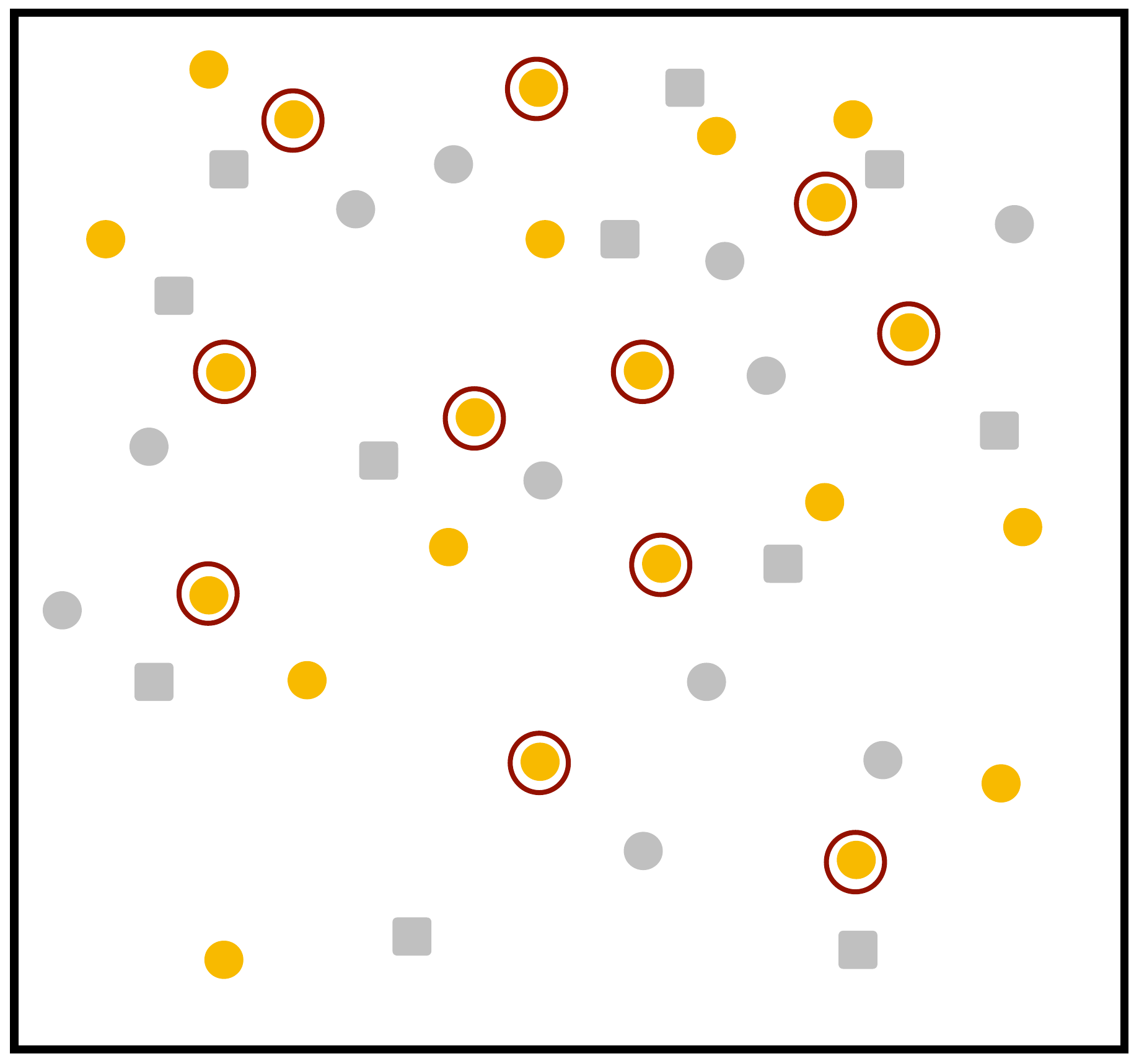}}
	\subfloat[{\scriptsize }\label{Figure: Sub-teri4}]{%
		\includegraphics[width=0.25\linewidth]{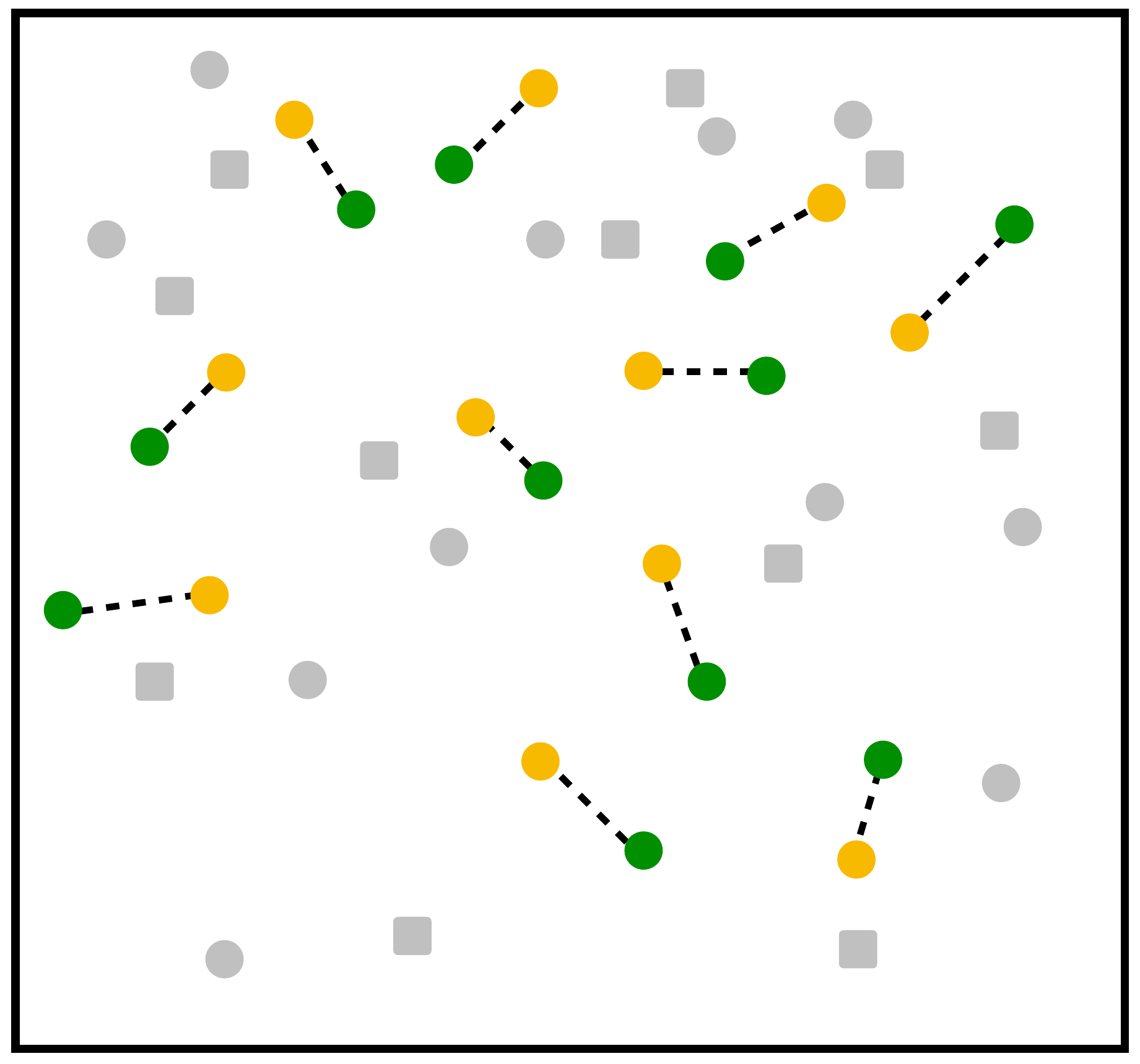}}
	\end{minipage}
	\caption{The mmWave AP network with multihop transmission protocol: (a) the whole AP network including backhaul-connect APs (tier 0) at blue squares, relay APs with hop count 1 (tier 1) at yellow dots, relay APs with hop count 2 (tier 2) at green dots; (b) transmission hop 1 from tier 0 to tier 1; (c) a subset of tier 1 (highlight in red circle) scheduled to transmit at hop 2; (d) transmission hop 2 from tier 1 to tier 2.}
	\label{Figure: sub-tier structure} 
	\vspace*{-15pt}
\end{figure*}
The mmWave APs have an inherently tiered structure due to the multihop protocol operated in the network, as shown in Fig.\ref{Figure: sub-tier structure}. Therefore, instead of modeling the whole network, we model the spatial distribution of mmWave APs for each transmission hop. The APs directly connected to the backhaul are defined as the $0^{\text{th}}$ tier. The other mmWave APs are used as relay, which can be further divided into multiple tiers with respect to 'the distance' to the $0^{\text{th}}$ tier. 

In the mmWave AP network, we measure the distance of an AP to tier 0 by the number of relays between the AP and tier $0$, which is termed as hop count. Based on the hop count, the $i^{\text{th}}$ AP tier or tier $i$ is defined as the subset of mmWave APs with hop count equal to $i$. 

Consider a mmWave AP network of density $\Lambda_{\text{A}}$, where the transmission protocol contains $M$ hops. It follows that mmWave APs can be divided into $M+1$ tiers, namely one backhaul-connected tier and $M$ relay tiers.
The locations of mmWave APs in tier $i$ are modeled by the point process $\Phi_i = \{x_1,x_2,\cdots\}$, $x_j \in \mathbb{R}^2$. 
We then define the sequence of point processes $\{\Phi_i\}_{i=0}^M$ as the topology of the mmWave AP network.

Note that the locations of backhaul-connected mmWave APs are restricted by the network infrastructure. Thus, we model tier 0 by a homogeneous Poison point process (PPP) $\Phi_0$ with intensity $\Lambda_0$. 
Unlike tier 0, the spatial distribution of $\Phi_{i+1}$ is determined by the point process $\Phi_i$ and the transmission protocol applied in the mmWave AP network. 
Denote $\phi_i$ as the subset of $\Phi_i$, where $\phi_i$ consists of the APs which are scheduled to transmit at hop $i+1$.
For the $i+1^{\text{th}}$ transmission hop, the AP located at $x \in \phi_i$ transmits to a cluster of points $\mathcal{B}_i^x =\{y_1,y_2,\cdots\}$, where $\mathcal{B}_i^x$ is centered at $x$.
It follows that tier $i+1$ can be expressed as $\Phi_{i+1}  = \bigcup_{x\in \phi_i}\, {\mathcal{B}_i^x}$.
Here, the points of $\mathcal{B}_i^x$ are assumed to be independently and identically distributed (i.i.d.) around the cluster center $x$.

\section{Topology of  mmWave AP network}

Assume that the hybrid precoding is implemented at the mmWave APs, which consists of analog beamforming and spatial multiplexing. Note that analog beamforming is mandatory for a mmWave AP to combat the high path loss. However, spatial multiplexing is required only when the mmWave AP needs to support multiple data streams. 

\begin{defn}{\textbf{Spatial Multiplexing Gain.}}
	The spatial multiplexing gain for a mmWave AP is defined as the number of data streams supported by the AP. 
\end{defn}

Assume the mmWave AP to be equipped with $K$ RF chains.
Then the spatial multiplexing gain of the mmWave AP is upper bounded by $K$.
Let all the transmitters of the same hop employ the identical hybrid precoding process. It follows that all APs in $\phi_i$ have the same spatial multiplexing gain, which is denoted by $k_i$.
Next, we characterize the topology of mmWave AP network with respect to the spatial multiplexing gain $k_i$.

\subsection{Topology with Spatial Multiplexing Disabled}
Given that the spatial multiplexing is disabled, we then have $k_i = 1, \forall i$.
It implies that for each $x \in \phi_i$, the cluster $\mathcal{B}_i^x$ contains only one point $y \in \Phi_{i+1}$, where the probability density function (PDF) of $y$ conditioning on $x$ is denoted by $f_i(y|x)$.
Assume that the mmWave APs in $\phi_i$ transmit at the same power.
Let each AP of tier $i+1$ be deployed to receive the maximum average power from $\phi_i$.
Given that $\phi_i$ is a PPP, we then derive the conditional PDF $f_i(y|x)$. 
\begin{lem}
In a multihop mmWave AP network, assume $x \in \phi_i$ and $y \in \Phi_{i+1}$ to be a pair of transmitter and receiver at the $i+1^{\text{th}}$ hop.	
If $\phi_i$ is a PPP with intensity $\lambda_i$, then $y$ is isotropically distributed around $x$ with the conditional PDF:
	\begin{eqnarray}
      f_i(y|x) &= &
      f_{\text{L}}(y|x) e^{-2 \pi {\lambda_i} \int_{0}^{|y-x|^{\alpha_{\text{L}}/\alpha_{\text{N}}}}P_{\text{N}}(r)r \, \text{d}r} 
      \;\;+ \;  f_{\text{N}}(y|x)e^{-2 \pi {\lambda_i} \int_{0}^{|y-x|^{\alpha_{\text{N}}/\alpha_{\text{L}}}}P_{\text{L}}(r)r \, \text{d}r},  \label{eq: probability of connecting to distance z} 
   \end{eqnarray}
   where
   \begin{eqnarray}
	  f_{\text{L}}(y|x) = 2 \pi |y-x| {\lambda_i} P_{\text{L}}(|y-x|) e^{-2 \pi {\lambda_i} \int_{0}^{|y-x|}P_{\text{L}}(r)rdr} \; , \nonumber\\
	  f_{\text{N}}(y|x) = 2 \pi |y-x| {\lambda_i} P_{\text{N}}(|y-x|) e^{-2 \pi {\lambda_i} \int_{0}^{|y-x|}P_{\text{N}}(r)rdr}. \nonumber
	\end{eqnarray}
	Here, the constant $\alpha_{\text{L}}$ and $\alpha_{\text{N}}$ represent the path loss exponent for line-of-sight (LOS) and non-line-of-sight (NLOS) mmWave link, respectively. $P_{\text{L}}(r)$ refers to the probability with that a mmWave link of length $r$ is LOS. It follows the NLOS probability $P_{\text{N}}(r) = 1 - P_{\text{L}}(r)$.
	\label{Lemma - sub-tier with ki=1}
\end{lem}
 \begin{proof}
	We start to derive the PDF of the distance from $y \in \Phi_{i+1}$ to its nearest LOS AP in $\phi_{i}$. Without loss of generality, let $y$ be the origin of the coordinate system. Denote the disc of radius $z$ centered at the origin as $\mathcal{D}^z_o$. Following the network model, the APs in $\phi_{i}$ which are LOS to the origin form an inhomogeneous PPP $\phi_\text{iL}$ with density $\lambda_{\text{iL}}(z) = \lambda_{\text{i}} P_{\text{L}}(z)$.  Thus the null probability of $\phi_{\text{iL}}$ in $\mathcal{D}^z_o$ is given by $\mathbb{P}(\phi_{\text{iL}} \cap \mathcal{D}^z_o=\emptyset) = e^{-2\pi \int_{0}^{z} {x \lambda_{\text{iL}}(r)} \text{d}r}$. 
	The PDF of $z$ for $\phi_{\text{iL}}$ can then be derived as
	\begin{equation}
	f_{\text{L}}(z) = \frac{\text{d} (1 - \mathbb{P}(\phi_{\text{iL}} \cap \mathcal{D}^z_o=\emptyset)) }{\text{d} z } = 2 \pi z \lambda_i P_{\text{L}}(z) e^{-2 \pi \lambda_i \int_{0}^{z}P_{\text{L}}(r)rdr}.
	\label{proofeq: f_L(z)}
	\end{equation}
	The PDF of $z$ for NLOS APs in $\phi_{i}$ i.e. $\phi_{\text{iN}}$ can be obtained by the same steps, where
    \begin{equation}
	f_{\text{N}}(z) = \frac{\text{d} (1 - \mathbb{P}(\phi_{\text{iN}} \cap \mathcal{D}^z_o=\emptyset)) }{\text{d} z } = 2 \pi z \lambda_i P_{\text{N}}(z) e^{-2 \pi \lambda_i \int_{0}^{z}P_{\text{N}}(r)rdr}.
	\label{proofeq: f_N(z)}
	\end{equation}
	
	Following the Bayes' rule, we have $\mathbb{P}\left( x_0 \in \phi_{\text{iL}} \;\boldsymbol{|}\; |x_0| = z \right) \propto {f( |x_0| = z \; \boldsymbol{|} \; x_0\in\phi_{\text{iL}}) \mathbb{P}(x_0 \in \phi_{\text{iL}})}$ and $\mathbb{P}\left( x_0 \in \phi_{\text{iN}} \;\boldsymbol{|}\; |x_0| = z \right) = {f( |x_0| \propto z \; \boldsymbol{|} \; x_0\in\phi_{\text{iN}}) \mathbb{P}(x_0 \in \phi_{\text{iN}})}$. 
	The  deductions of $\mathbb{P}(x_0 \in \phi_{\text{iL}})$ and $\mathbb{P}(x_0 \in \phi_{\text{iN}})$ follow the similar steps in \cite[Theorem 2]{bai2015coverage}. In \cite[Theorem 3]{bai2015coverage}, $f( |x_0| = z \; \boldsymbol{|} \; x_0\in\phi_{\text{iL}}) $ and $f( |x_0| = z \; \boldsymbol{|} \; x_0\in\phi_{\text{iN}}) $ are derived as a function of $f_\text{L}(z)$ in (\ref{proofeq: f_L(z)}) and $f_\text{N}(z)$ in (\ref{proofeq: f_N(z)}). 
\end{proof}
Lemma \ref{Lemma - sub-tier with ki=1} implies that each point $y \in \Phi_{i+1}$ can be considered as the isotropic displacement of some point $x \in \phi_i$, where the distance between $y$ and $x$ follows the PDF $f_i(y|x)$ in (\ref{eq: probability of connecting to distance z}). It follows from \cite{haenggi2012stochastic} that  if $\phi_i$ is a homogeneous PPP, then $\Phi_{i+1}$ is a PPP with the same intensity of $\phi_{i}$. 

Consider a mmWave AP network of density $\Lambda_{\text{A}}$, where all APs are scheduled to transmit to the following tier i.e. $\phi_i = \Phi_i$, $\forall i$. By repeatedly using the displacement property of the PPP \cite{haenggi2012stochastic}, the topology of mmWave AP network can then be written as $\{\Phi_i\}_{i=0}^M$, where $\Phi_{i}$ is a homogeneous PPP of intensity $\Lambda_0$, $\forall i$. It follows that the total number of transmission hops $M = \frac{\Lambda_{\text{A}}}{\Lambda_0}-1$.
Note that $\bigcup_{i=0}^M \Phi_i$ is the superposition of $M+1$ homogeneous PPPs, thus is also a homogeneous PPP. It implies that the mmWave APs are uniformly distributed in the network if the spatial multiplexing at APs is disabled by the transmission protocol.

\begin{figure}
	\vspace*{-15pt}
	\centering
	\subfloat[{\scriptsize The total number of transmission hops equal to $M = 12$. The spatial multiplexing gain for each hop $k_i = 1$.} \label{Figure: AP topology with ki = 1}]{%
		\includegraphics[width=0.48\linewidth]{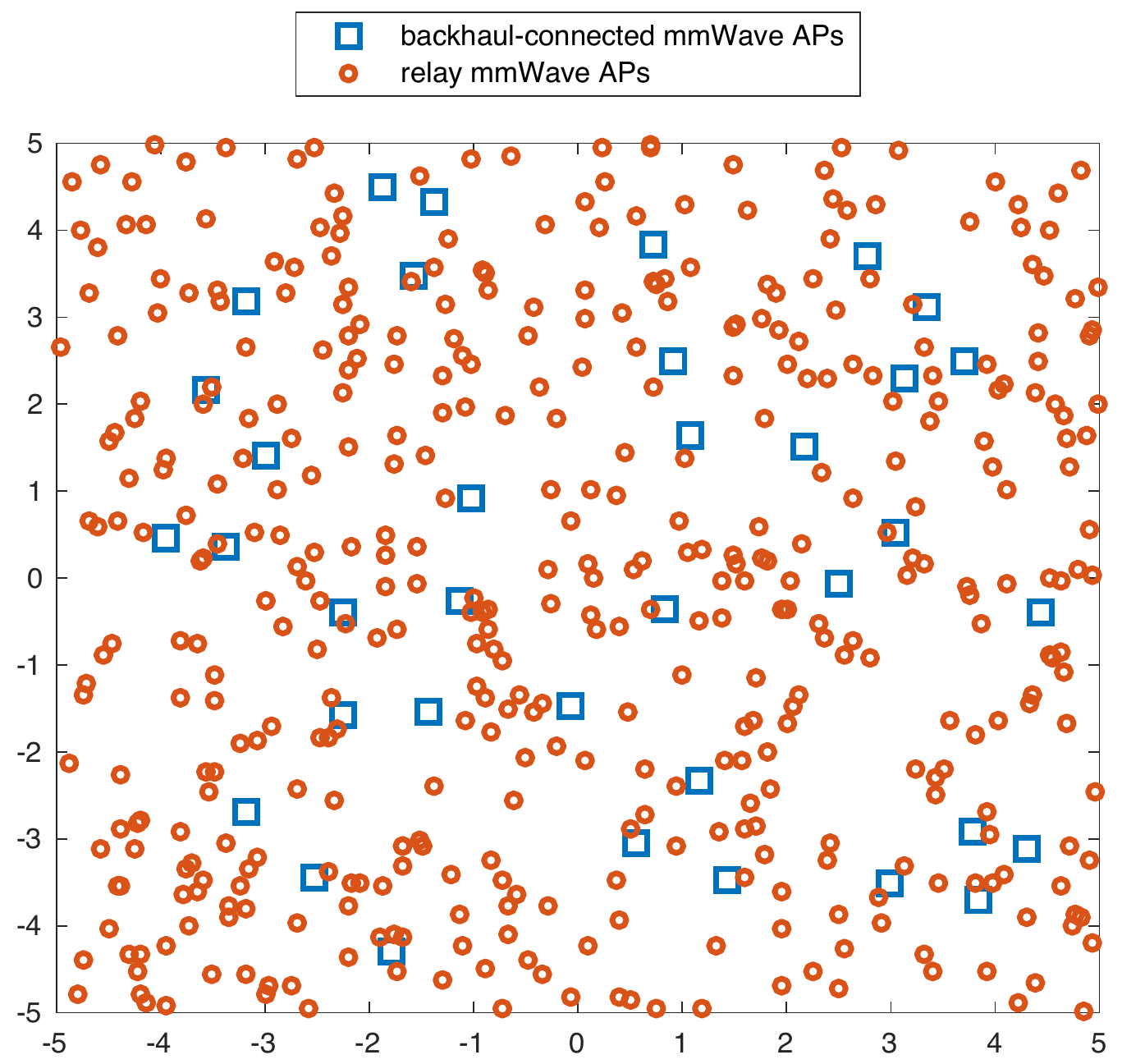}}
	\hfill
	\subfloat[{\scriptsize The total number of transmission hops equal to $M = 2$. The spatial multiplexing gain for each hop $k_i = 6$.}\label{Figure: AP topology with ki = 6}]{%
		\includegraphics[width=0.48\linewidth]{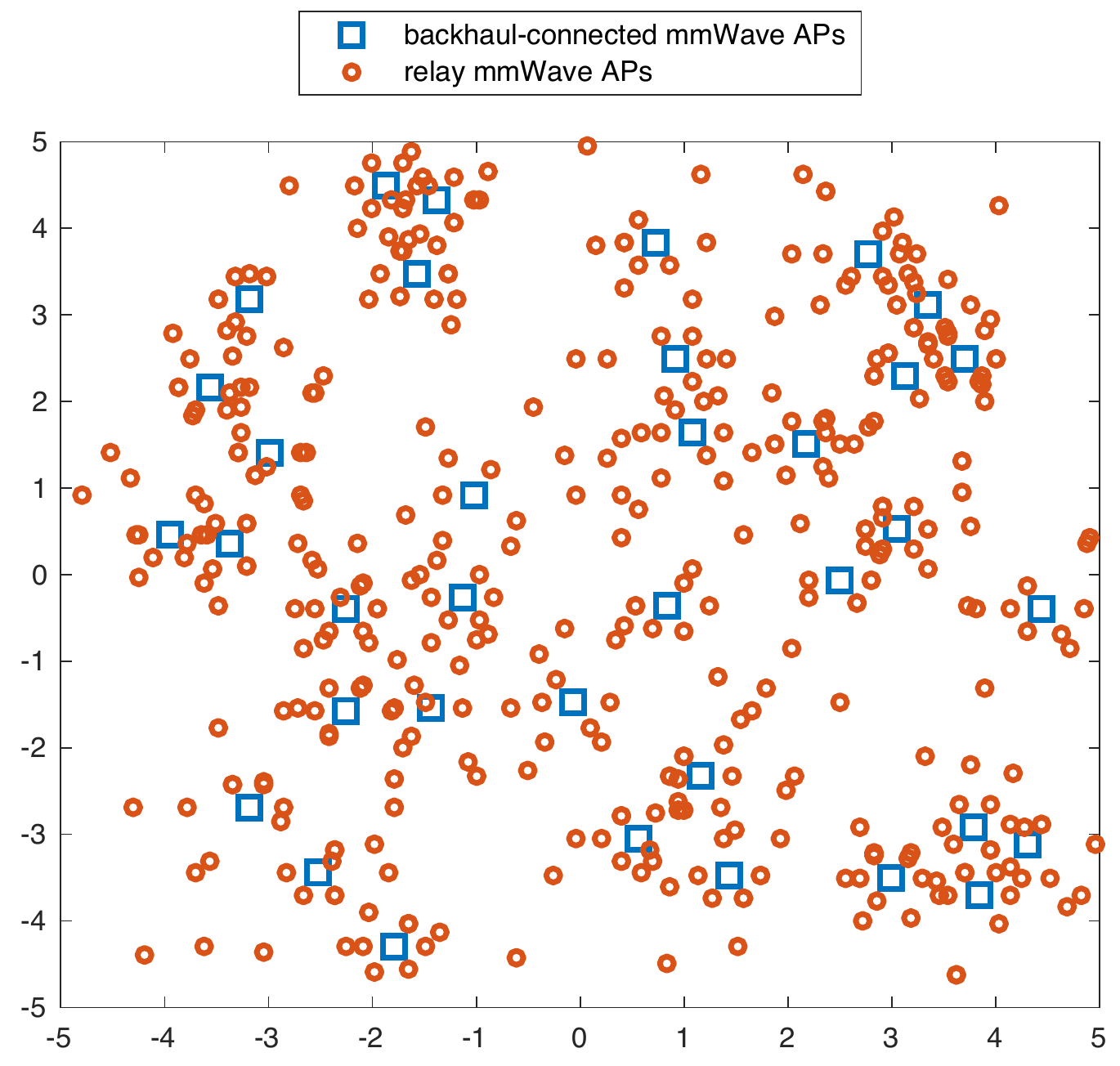}}
	\caption{The spatial distribution of mmWave APs in the network with spatial multiplexing disabled (left) and enabled (right).}
	\label{Figure: AP topology} 
	 \vspace*{-15pt}
\end{figure}
	
\subsection{Topology with Spatial Multiplexing Enabled}
By enabling the spatial multiplexing of APs in $\phi_i$, the cluster $\mathcal{B}_i^x, x\in\phi_{i}$ becomes a sequence of i.i.d. points $\{y_1,\cdots,y_{k_i}\}$ with the PDF $f_i(y_j|x)$. 
Following Lemma \ref{Lemma - sub-tier with ki=1}, if $\phi_i$ is a PPP of intensity $\lambda_i$, then $y_j$ is isotropically located around $x$ with $f_i(y_j|x)$ in (\ref{eq: probability of connecting to distance z}).
We remark that the location distribution of $y_j$ only depends on $x$, which is irrelevant to the size of $\mathcal{B}_i^x$. 

For tier $i+1$, the mmWave APs are spatially distributed following the point process $\Phi_{i+1} = \bigcup_{x\in \phi_i}\;\mathcal{B}_i^x$. Given that $\phi_i$ is a PPP of intensity $\lambda_i$, we then have that $\Phi_{i+1}$ is a Poisson cluster process (PCP), more specifically, a Neyman-Scott process with intensity $\Lambda_{i+1} = k_i \lambda_i$ \cite{haenggi2012stochastic}.

Note that $\Phi_{i}$ is a collection of clusters $\mathcal{B}_{i-1}^x$ with size $k_{i-1}$. Let $\phi_i$ be generated by taking one point from each cluster $\mathcal{B}_{i-1}^x, x\in\phi_{i-1}$. It follows from Lemma \ref{Lemma - sub-tier with ki=1} that if $\phi_{i-1}$ is a PPP, then $\phi_{i}$ is a PPP, where $\phi_{i}$ and $\phi_{i-1}$ are of the same intensity. Note that $\Phi_0$ represents the transmitters at hop 1, which is a PPP of intensity $\Lambda_0$. Therefore, each $\phi_{i}$ is a PPP with intensity $\Lambda_0$. It follows that $\Phi_{i+1}$ is a Neyman-Scott point process with intensity $ k_i\Lambda_0$ for all $i$.

In Fig.\ref{Figure: AP topology with ki = 6}, we illustrate the spatial distribution of mmWave APs with spatial multiplexing enabled. Note that the mmWave AP networks in Fig.\ref{Figure: AP topology with ki = 1} and Fig.\ref{Figure: AP topology with ki = 6} are of the same density $\Lambda_{\text{A}}$. Moreover, tier 0 is assumed to be identical for the two AP networks. However, the spatial multiplexing is disabled in Fig.\ref{Figure: AP topology with ki = 1}. It can be observed from Fig.\ref{Figure: AP topology} that given $k_i = 1, \forall i$, the mmWave relays are uniformly distributed in the area. By enabling the spatial multiplexing, the locations of mmWave relays become clustering. Such a clustering pattern is consistent with the distribution of Neyman-Scott point process.

\section{Performance of mmWave AP Network}
To characterize the impact of spatial multiplexing on the throughput, we first derive the latency and coverage probability for the mmWave AP network. Note that the latency of mmWave AP network indicates the delay of packet transmission. While the reliability of packet transmission can be captured by the coverage probability.

 \subsection{Network Latency}
In a mmWave AP network, tier $0$ is always the source tier for a packet, whereas the destination of the packet can be a mmWave AP in any tier. 
Note that the delay of a packet depends on the hop count from the source AP to the destination AP. Accordingly, the worst-case delay of a packet is equivalent to the maximum hop count of APs in the network. It follows that the latency of a mmWave AP network can be defined as the total number of transmission hops $M$ in the network.
\begin{thm}
	For a mmWave AP network of density $\Lambda_{\text{A}}$, the latency $M$ is bounded by 
	\begin{equation}
	 \frac{\Lambda_{\text{A}}}{K\Lambda_0} - \frac{1}{K} \leq M \leq \frac{\Lambda_{\text{A}}}{\Lambda_0} - 1,
	\end{equation}
	where $\Lambda_0$ is the intensity of tier 0. Each mmWave AP is equipped with $K$ RF chains, thus $K$ represents the maximum spatial multiplexing gain for each transmission hop.
\label{thm: network latency}	
\end{thm}

\begin{proof}
	The network latency $M$ satisfies
	$\sum_{i=1}^{M} \Lambda_{i} = \Lambda_{\text{A}} - \Lambda_0$. Note that $ \Lambda_0 \leq \Lambda_{i} \leq K\Lambda_0$, the result immediately follows.
\end{proof}

Theorem \ref{thm: network latency} demonstrates that the latency of mmWave AP network decreases linearly with the increase of spatial multiplexing gain. Note that the network latency reaches its upper bound when the spatial multiplexing is disabled. By setting the spatial multiplexing gain to $K$ for each transmission hop, the minimum latency of mmWave AP network can be achieved.

\subsection{Coverage probability}
To calculate the coverage probability for tier $i+1$, we need to first derive the signal-to-interference-noise ratio (SINR) for APs in $\Phi_{i+1}$. Note that $\phi_{i}$ and $\Phi_{i+1}$ represent the transmitters and receivers of the $i+1^{\text{th}}$ transmission hop, respectively.
At hop $i+1$, the hybrid precoding is employed by two stages \cite{andrews2017modeling}. 
In the first stage, the mmWave AP located at $x \in \phi_{i}$ assigns a unique analog beam to each AP in $\mathcal{B}_i^x$. Denote $\theta_{\text{A}}$ as the main lobe width of the analog beam, where the beamforming gains within and outside the main lobe are denoted by $G_{\text{A}}$ and $g_{\text{A}}$, respectively. We use $\mathbb{G}(k_i)$ to denote the transceiver beamforming gain between two APs.
It follows that $\mathbb{G}(k_i)$ equals to $G_{\text{A}}^2$, $G_{\text{A}}g_{\text{A}}$ and $g_{\text{A}}^2$ with probabilities $\left( \frac{\theta_{\text{A}}k_i}{2\pi} \right)^2$, $\frac{\theta_{\text{A}}k_i}{\pi} \left( 1 - \frac{\theta_{\text{A}}k_i}{2\pi} \right) $ and $ \left( 1 - \frac{\theta_{\text{A}}k_i}{2\pi} \right)^2$, respectively. 

In the second stage of hybrid precoding, the spatial multiplexing is performed in the baseband, where the AP at $x \in \phi_{i}$ transmits a different data stream for each $y \in \mathcal{B}_i^x$ as well as cancels the inter-stream interference. 
Let a randomly selected AP at $y \in \Phi_{i+1}$ be the origin of the coordinate system, where $y$ belongs to the cluster $\mathcal{B}_i^x$. It follows that the coordinate of $x$ is translated to $x_0 = x - y$. The SINR of the AP at the origin can then be expressed as
\begin{eqnarray}
\text{SINR}(k_i) \triangleq \frac{h_0 G_{\text{A}}^2 \ell(|x_0|)}{\sigma^2 + \mathcal{I}_i(k_i)}  
 = \frac{h_0 G_{\text{A}}^2 \ell(|x_0|)}{\sigma^2 + \sum_{x_b\in\phi_{i}\setminus\{x_0\}}h_b\mathbb{G}_b(k_i)\ell(|x_b|)},
\label{def: SINR distribution}
\end{eqnarray}
where $h_b$ represents the channel fading from $x_b$ to the origin; $\mathbb{G}_b(k_i)$ is the transceiver beamforming gain between $x_b$ and the origin; $\sigma^2$ denotes the noise power; $\ell(\cdot)$ denotes the path loss of the mmWave link \cite{andrews2017modeling} 
\begin{equation}
\ell(r) = 
\left\lbrace 
\begin{split}
& \beta r^{-\alpha_{\text{L}}}, \text{with probability}\;P_{\text{L}}(r)\\
& \beta r^{-\alpha_{\text{N}}}, \text{with probability}\;P_{\text{N}}(r)
\end{split},
\right.
\label{eq: path loss model} 
\end{equation}
where $\beta$ is a constant representing the intercept of path loss model \cite{andrews2017modeling}; the LOS probability $P_{\text{L}}(r)$ and NLOS probability $P_{\text{N}}(r)$ are introduced in Lemma \ref{Lemma - sub-tier with ki=1}.

As discussed in Section III, $\phi_{i}$ is formed by taking one point from each cluster of $\Phi_{i}$. Thus, $\phi_{i}$ is always a PPP with intensity $\Lambda_0$ regardless of the intensity of $\Phi_{i}$. It means that the SINR distribution in (\ref{def: SINR distribution}) depends only on $k_i$ and $\tau$.
Based on (\ref{def: SINR distribution}), the coverage probability of tier $i+1$ is defined as 
\begin{equation}
\mathcal{C}(\tau,k_i) \triangleq \mathbb{P}(\text{SINR}(k_i)>\tau)
\label{def: P_cov}
\end{equation}
with that the AP in tier $i+1$ has a SINR larger than the threshold $\tau$.
To calculate the coverage probability, we first provide the characteristic function of the interference at the origin. Given that the AP at the origin is connected to a LOS AP located at $x_0$, the characteristic function can be written as
\begin{eqnarray}
\mathcal{L}_{\mathcal{I}_\text{L}}(s,k_i) 
&= & \exp \left( - 2\pi {\Lambda_0} \int_{|x_0|}^{\infty} \left[
1-\mathcal{G}_{\text{L}}(s,r,k_i)
\right]  P_{\text{L}}(r)r \; \text{d}r \right) \nonumber\\
&&\;\;\times  \exp \left( - 2\pi {\Lambda_0} \int_{|x_0|^{\alpha_{\text{L}}/ \alpha_{\text{N}}}}^{\infty} \left[
1-\mathcal{G}_{\text{N}}(s,r,k_i)
\right]   P_{\text{N}}(r)r \; \text{d}r \right), 
\label{eq: Laplace transform of LOS}
\end{eqnarray}
where $s$ is the value on that the characteristic function is evaluated and $\mathcal{G}_{\text{L}}(s,r,k_i) = \mathbb{E}_{h,\mathbb{G}} 
\left[
e^{- s \beta h \mathbb{G}(k_i) r^{-\alpha_{\text{L}}}} 
\right]$.

If the AP at $x_0$ is in the NLOS state, the characteristic function of the interference is given as
\begin{eqnarray}
\mathcal{L}_{\mathcal{I}_\text{N}}(s, k_i) 
&=& \exp \left( - 2\pi {\Lambda_0} \int_{|x_0|}^{\infty} \left[
1-\mathcal{G}_{\text{N}}(s,r,k_i)
\right]   P_{\text{N}}(r)r \; \text{d}r \right) \nonumber\\
&&\;\;\times  \exp \left( - 2\pi {\Lambda_0} \int_{{|x_0|}^{\alpha_{\text{N}}/ \alpha_{\text{L}}}}^{\infty} \left[
1-\mathcal{G}_{\text{L}}(s,r,k_i)
\right]  P_{\text{L}}(r)r \; \text{d}r \right)
\label{eq: Laplace transform of NLOS}
\end{eqnarray}
\label{lem: Laplace Transform}
with
$\mathcal{G}_{\text{N}}(s,r,k_i) = \mathbb{E}_{h,\mathbb{G}} 
\left[
e^{- s \beta h \mathbb{G}(k_i) r^{-\alpha_{\text{N}}}} 
\right]$.

Next, we show the main result on the coverage probability of tier $i+1$ for a mmWave AP network.
\begin{thm}
	In a mmWave network, the coverage probability of a randomly selected AP in tier $i+1$ is given by
	\begin{eqnarray}
	\mathcal{C}(\tau,k_i) 
	&= &\int\limits_{r>0}
	\frac{ f_{\text{L}}(r) f_{\text{N}}(r^{\alpha_{\text{L}}/\alpha_{\text{N}}})}
	{2\pi r^{\alpha_{\text{L}}/\alpha_{\text{N}}} {\Lambda_0}
		 P_{\text{N}}(r^{\alpha_{\text{L}}/\alpha_{\text{N}}})}
	\mathcal{C}_{\text{L}}(\tau,r,k_i)
	\; \text{d} r  \nonumber\\
	&&\quad + 
	\int\limits_{r>0}
	\frac{ f_{\text{L}}(r^{\alpha_{\text{N}}/\alpha_{\text{L}}}) f_{\text{N}}(r)}
	{2\pi r^{\alpha_{\text{N}}/\alpha_{\text{L}}} {\Lambda_0} P_{\text{L}}(r^{\alpha_{\text{N}}/\alpha_{\text{L}}})}  
	\mathcal{C}_{\text{N}}(\tau,r,k_i)
	\; \text{d} r, 
	\label{eq: SINR interms of conditional SINR}
	\end{eqnarray}
	where $k_i$ is the spatial multiplexing gain at hop $i$; $f_{\text{L}}(\cdot)$ and $f_{\text{N}}(\cdot)$ are the PDF of distance distributions given in (\ref{eq: probability of connecting to distance z}).  
	Assume that channel follows the Rayleigh fading model, then
	\begin{align}
	\mathcal{C}_{\text{L}}(\tau,r,k_i) &= \exp \left( - {\frac{r^{\alpha_{\text{L}}}\tau \sigma^2}{G_{\text{A}}^2 \beta}} \right)
	\mathcal{L}_{\mathcal{I}_\text{L}} \left( \frac{r^{\alpha_{\text{L}}}\tau}{G_{\text{A}}^2 \beta}, k_i \right),
	\nonumber\\
	\mathcal{C}_{\text{N}}(\tau,r,k_i) &= \exp \left( - {\frac{r^{\alpha_{\text{N}}}\tau \sigma^2}{G_{\text{A}}^2 \beta}} \right)
	\mathcal{L}_{\mathcal{I}_\text{N}} \left( {\frac{r^{\alpha_N}\tau}{G_{\text{A}}^2\beta}}, k_i \right).
	\nonumber
	\end{align}
	\label{thm: SINR coverage probability}
	\vspace*{-10pt}
\end{thm} 

\begin{proof}
	For the AP of $\Phi_{i+1}$ at the origin, the SINR coverage probability can be written as
	\begin{equation}
	\mathcal{C}(\tau,k_i) \triangleq \mathbb{P}(\text{SINR}(k_i)>\tau) = \mathbb{E}_{r}\left[\mathcal{C} (\tau,k_i,r)\right], \nonumber
	\end{equation}
	where $\mathcal{C} (\tau,k_i,r) \triangleq \mathbb{P}(\text{SINR}(k_i)>\tau\; \boldsymbol{|} \; |x_0| = r)$ is the conditional SINR coverage probability. 
	According to the LOS AP or NLOS AP at $x_0$, $\mathcal{C} (\tau,k_i,r)$ can be further expanded as
	\begin{eqnarray}
	\mathcal{C} (\tau,k_i,r) 
	=\mathbb{P}\left( x_0 \in \phi_{\text{iL}} \;\boldsymbol{|}\; |x_0| = r \right)  \mathcal{C}_{\text{L}} (\tau,k_i,r)
	+ \mathbb{P}\left( x_0 \in \phi_{\text{iN}} \; \boldsymbol{|}\;  |x_0| = r \right) \mathcal{C}_{\text{N}} (\tau,k_i,r),
	\nonumber 
	\end{eqnarray}
	where 
	\begin{eqnarray}
	\mathcal{C}_{\text{L}} (\tau,k_i,r)&\triangleq& \mathbb{P}\left( \text{SINR} > \tau \; \boldsymbol{|} \; |x_0| = r, x_0\in\phi_{\text{iL}} \right) 
	=  \mathbb{P}\left(h_0 > \frac{\tau(\sigma^2 + \mathcal{I}_\text{L})}{G_{\text{A}}G_{\text{U}} \beta r^{-\alpha_{\text{L}}}} \right) \nonumber
	\end{eqnarray}
	and
	\begin{eqnarray}
	\mathcal{C}_{\text{N}} (\tau,k_i,r)&\triangleq& \mathbb{P}\left( \text{SINR} > \tau \; \boldsymbol{|} \; |x_0| = r, x_0\in\phi_{\text{iN}} \right) 
	= \mathbb{P}\left( h_0 > \frac{\tau(\sigma^2 + \mathcal{I}_\text{N})}{G_{\text{A}}G_{\text{U}} \beta r^{-\alpha_{\text{N}}}} \right). \nonumber
	\end{eqnarray}
	The conditional probabilities $\mathcal{C}_{\text{L}} (\tau,k_i,r)$ and $\mathcal{C}_{\text{N}} (\tau,k_i,r)$ can be expressed as a function of the characteristic functions in (\ref{eq: Laplace transform of LOS}) and (\ref{eq: Laplace transform of NLOS}) \cite{haenggi2012stochastic}. The details in deriving the characteristic function is shown in \cite{andrews2011tractable}.
	In Lemma \ref{Lemma - sub-tier with ki=1}, we prove the expressions of conditional probabilities
	\begin{eqnarray}
	\mathbb{P}\left( x_0 \in \phi_{\text{iL}} \;\boldsymbol{|}\; |x_0| = r \right) =    f_{\text{L}}(r) e^{-2 \pi {\lambda_i} \int_{0}^{r^{\alpha_{\text{L}}/\alpha_{\text{N}}}}P_{\text{N}}(z)z \, \text{d}z} 
	\end{eqnarray}
	and
	\begin{equation}
	\mathbb{P}\left( x_0 \in \phi_{\text{iN}} \;\boldsymbol{|}\; |x_0| = r \right) =    f_{\text{N}}(r)e^{-2 \pi {\lambda_i} \int_{0}^{r^{\alpha_{\text{N}}/\alpha_{\text{L}}}}P_{\text{L}}(z)z \, \text{d}z},
	\end{equation}
	where $f_\text{L}(\cdot)$ and $f_\text{N}(\cdot)$ are also given in Lemma \ref{Lemma - sub-tier with ki=1}. The coverage probability $\mathcal{C}(\tau,k_i)$ then follows.
\end{proof}
In Fig.\ref{Figure: SINR coverage probability of sub-tier i+1}, we numerically evaluate (\ref{eq: SINR interms of conditional SINR}) by showing $\mathcal{C}(\tau, k_i)$ with respect to $k_i$. The main lobe width of the analog beam is set as $\theta_{\text{A}} = 30^\circ$ with the main lobe gain $G_{\text{A}} = 20$dB and side lobe gain $g_{\text{A}} = 0$dB. The noise power is assumed to be negligible. Each mmWave AP is assumed to be equipped with $K = 12$ RF chains, which indicates the spatial multiplexing gain $k_i \leq 12$. In Section III, we show that $\phi_{i}, \forall i$, is of the same intensity as tier 0. Here, $\Lambda_0$ is represented by the inter-AP distance $r_0 =\sqrt{ 1/\pi \Lambda_0}$ \cite{bai2015coverage}. For the path loss model, we use $\beta = 1$, $\alpha_{\text{L}} = 2$ and $\alpha_{\text{N}} = 4$ \cite{rappaport2013millimeter}. It can be observed from Fig.\ref{Figure: SINR coverage probability of sub-tier i+1} that given the SINR threshold $\tau$, the decrease of coverage probability $\mathcal{C}(\tau, k_i)$ is close to linear with the increase of $k_i$. The other observation is that $\phi_{i}$ with $r_0= 100$m and $r_0 = 200$m results in the similar coverage probability of tier $i+1$. Since a higher density of transmitters results in higher power of desired signal, but leads to higher strength of interference.

\begin{figure} 
	\vspace{-0.1cm}
	\begin{minipage}[t]{0.47\textwidth}
		\centering
		\includegraphics[width=\linewidth]{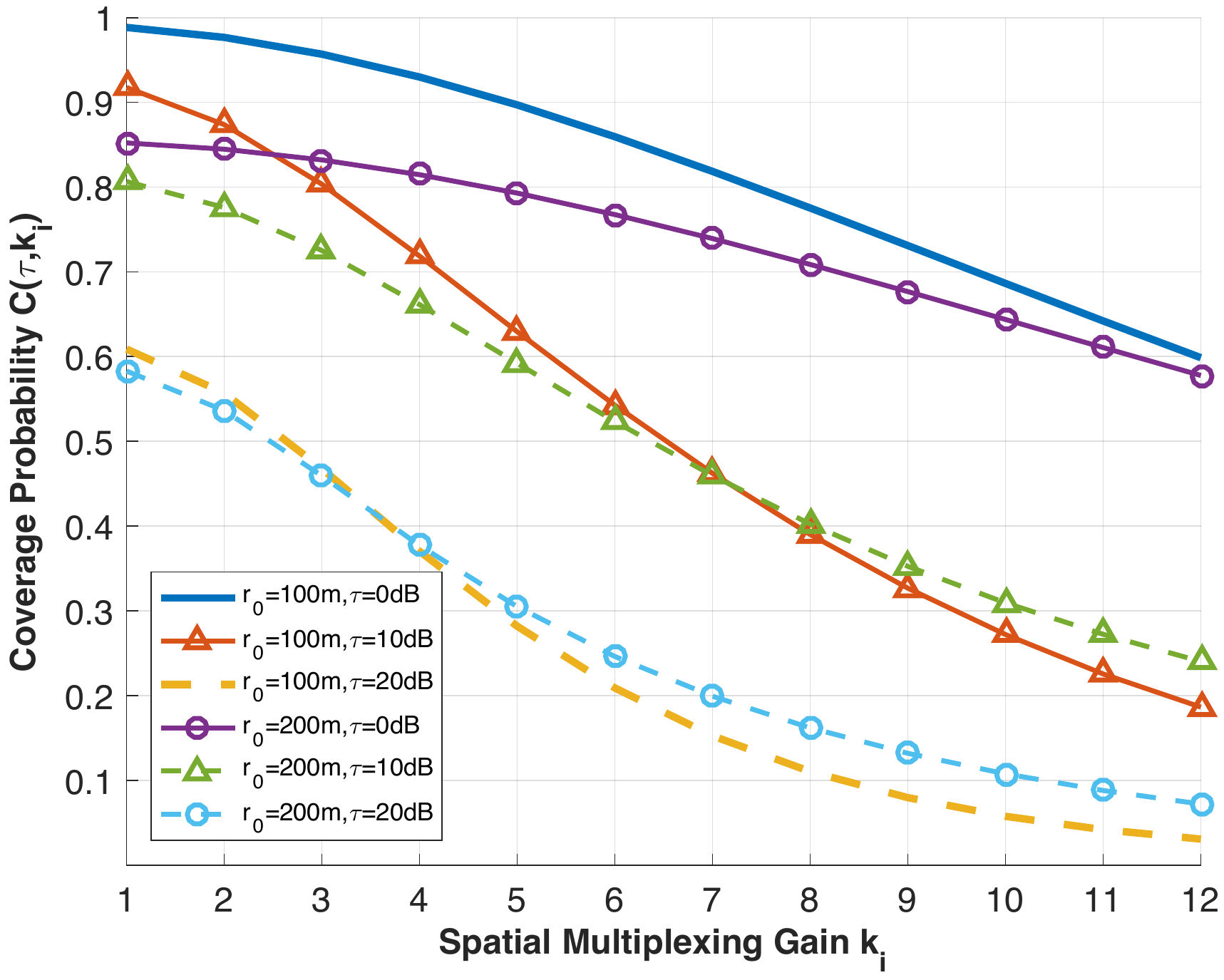}
		\caption{The coverage probability $\mathcal{C}(\tau, k_i)$ with different spatial multiplexing gain $k_i$. The intensity $\Lambda_0$ is measured by inter-APs distance $r_0$, where $r_0 = \sqrt{\frac{1}{\pi {\Lambda_0}}}$. Here, $\theta_{\text{A}} = 30^{\circ}$, $G_{\text{A}} = 20$dB and $g_{\text{A}} = 0$dB.}
		\label{Figure: SINR coverage probability of sub-tier i+1}
	\end{minipage} %
	\hfill
	\begin{minipage}[t]{0.47\textwidth}
		\includegraphics[width=\linewidth]{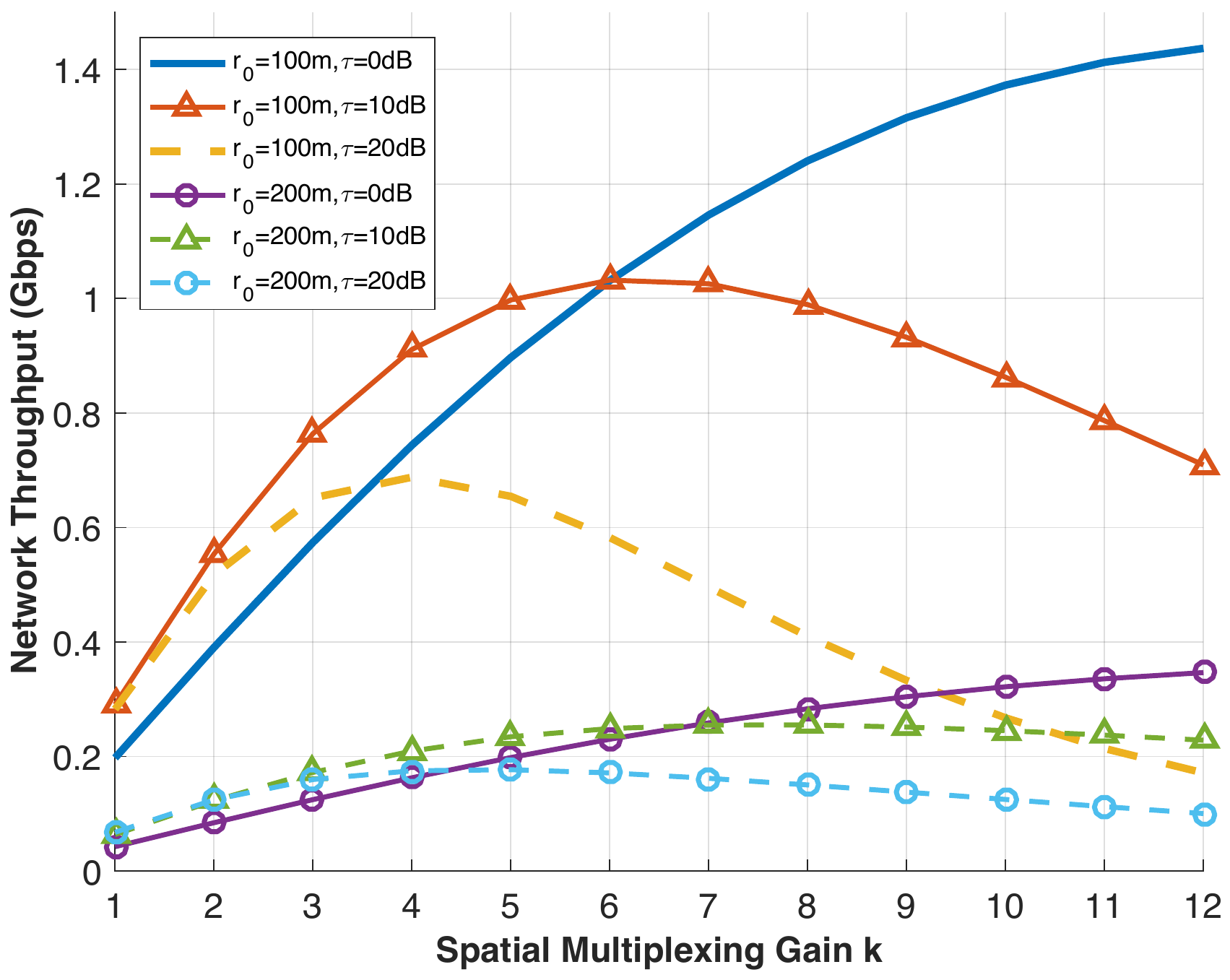}
		\caption{Network throughput with different spatial multiplexing gain $k$. The intensity $\Lambda_0$ is denoted by $r_0 = \sqrt{\frac{1}{\pi {\Lambda_0}}}$. Here, $\theta_{\text{A}} = 30^{\circ}$, $G_{\text{A}} = 20$dB and $g_{\text{A}} = 0$dB.}
		\label{Figure: network throughput}
	\end{minipage}
	\vspace{-20pt} 
\end{figure}

\subsection{Network Throughput}
For a mmWave AP network, we define the throughput as the aggregate data rate of all relay tiers i.e. $\bigcup_{i=1}^M {\Phi_{i}}$. 
Consider the network, where each transmission hop follows the same protocol. It follows that $k_i = k, \forall i$. Consequently, the coverage probability of each tier is equal to $\mathcal{C}(\tau, k)$. Assume that each AP in the network is assigned with the same bandwidth. We then derive the network throughput with respect to spatial multiplexing gain.
\begin{thm}
	Consider a mmWave AP network of density $\Lambda_{\text{A}}$, where the distribution of backhaul-connected APs follows a homogeneous PPP with intensity $\Lambda_0$. Assume that each hop follows the identical transmission protocol with the spatial multiplexing gain $k$, the network throughput is then given by
	\begin{align}
	\mathbb{T}(k) &\triangleq \frac{W(\Lambda_{\text{A}} - \Lambda_0)\mathcal{C}(\tau, k)\log_2(1+\tau)}{M}\nonumber\\ &=Wk\Lambda_0\mathcal{C}(\tau, k)\log_2(1+\tau),
	\end{align}
	where $\tau$ is the SINR threshold; $M = (\Lambda_{\text{A}} - \Lambda_0)/k\Lambda_0$ is the network latency; $W$ denotes the bandwidth of each AP; the coverage probability $\mathcal{C}(\tau, k)$ is given in (\ref{eq: SINR interms of conditional SINR}).
	\label{thm: network throughput}
\end{thm}
\begin{proof}
	By combining Theorem \ref{thm: network latency} and \ref{thm: SINR coverage probability}, the result immediately follows.
\end{proof}
We remark that the throughput of a mmWave AP network $\mathbb{T}(k)$ is dependent on the intensity of backhaul-connected APs i.e. $\Lambda_0$. However, $\mathbb{T}(k)$ is independent of $\Lambda_{\text{A}}$. 

In Fig.\ref{Figure: network throughput}, we show the network throughput derived in Theorem \ref{thm: network throughput}. The parameters used in Fig.\ref{Figure: SINR coverage probability of sub-tier i+1} and Fig.\ref{Figure: network throughput} are the same. 
For a fixed intensity $\Lambda_0$, Fig.\ref{Figure: network throughput} demonstrates the trade-off between the latency and coverage probability when choosing the spatial multiplexing gain $k$. It can be observed from Fig.\ref{Figure: network throughput} that the network throughput is not monotonically varying with the spatial multiplexing gain. As the network latency can be improved by increasing $k$, whereas the degradation of coverage probability is considerable as $k$ increases. 
Moreover, the optimal spatial multiplexing gain is different depending on the SINR threshold. When the SINR threshold is low, the latency dominates the network throughput. Consequently, the higher spatial multiplexing gain corresponds to the higher throughput. In the region of the high SINR threshold, the coverage probability becomes the major challenge of network performance. Therefore, the smaller spatial multiplexing gain results in the higher throughput for the mmWave AP network.
By comparing the throughputs of networks with different $\Lambda_0$, Fig.\ref{Figure: network throughput} then illustrates that the ultra-dense mmWave AP network can benefit from the densification of tier 0. 

\section{Conclusion}
We proposed to incorporate the multihop transmission protocol in modeling the ultra-dense mmWave AP network. 
Moreover, we exploited the spatial distributions of mmWave APs with different transmission protocols.
Our analysis indicated that the mmWave AP network can be modeled by a PPP if spatial multiplexing is disabled. However, the topology of mmWave AP network is a collection of Neyman-Scott processes when the spatial multiplexing is enabled.
We then analyzed the performance for mmWave AP networks with different topologies. We showed that the uniform distribution of mmWave APs experiences the largest latency but has the highest coverage probability.
Moreover, the latency of mmWave AP network decreases as the spatial multiplexing gain increases, while the coverage probability drops with the increase of spatial multiplexing gain. The numerical results showed the optimal spatial multiplexing gain to maximize the throughput of the ultra-dense mmWave AP network.

\small
\bibliographystyle{IEEEtran}
\bibliography{mmWavepapers.bib}
\end{document}